\documentclass[final]{eptcs}
\usepackage[utf8]{inputenc}

\usepackage{bussproofs}
\EnableBpAbbreviations 

\usepackage{etoolbox} 

\usepackage{appendix}
\usepackage{xspace}

\newtoggle{eptcs}
\toggletrue{eptcs}

\usepackage{amsthm}

\usepackage{amsmath}
\usepackage{amssymb}
\usepackage{stmaryrd}

\usepackage{graphicx}
\usepackage{marvosym}

\usepackage{color}
\usepackage[usenames,dvipsnames]{xcolor}

\usepackage{pdfsync}
\usepackage{wasysym}
\usepackage{xifthen} 
\usepackage{mathtools}

\usepackage{alltt} 

\usepackage{paralist} 

\usepackage{fixltx2e} 
\usepackage{xspace} 

\usepackage{nicefrac}
\usepackage{environ,thmtools}  

\graphicspath{{figures/}} \DeclareGraphicsExtensions{.pdf,.jpg}

\usepackage{hyperref} 
\hypersetup{final}
\usepackage{cleveref}

\usepackage[final,nomargin,inline,index]{fixme} 
\fxusetheme{color}

\newcommand{\codefont}{\fontsize{9}{9}\selectfont}
\newcommand{\code}[1]{{\tt\codefont {#1}}}

\newcommand{\ifempty}[3]{%
  \ifthenelse{\isempty{#1}}{#2}{#3}%
}

\newcommand{\ifzero}[3]{%
  \ifthenelse{\equal{#1}{0}}{#2}{#3}%
}

\def\eg{e.g.\@\xspace}
\def\ie{i.e.\@\xspace}

\newtheorem{theorem}{Theorem}
\newtheorem{definition}{Definition}
\newtheorem{lemma}{Lemma}
\newtheorem{example}{Example}

\newtheorem{proposition}{Proposition}

  {\end{myproof}}

\newcommand{\Realpos}{\mathbb{R}_{\geq 0}}

\newcommand{\Nat}{\mathbb{N}}

\newcommand{\reset}[2]{{#1}[{#2}]}

\newcommand{\past}[1]{\mathbin{\downarrow {#1}}}

\newcommand{\rdy}[1]{\code{rdy}\ifempty{#1}{}{({#1})}}

\def\atomColor{\color{Magenta}}
\newcommand{\atomFmt}[1]{{\atomColor{\code{#1}}}}
\newcommand{\bang}{{\textup{\texttt{\symbol{`\!}}}}}
\newcommand{\qmark}{{\textup{\texttt{\symbol{`\?}}}}}

\newcommand{\InTag}{\atomFmt{\qmark}}
\newcommand{\OutTag}{\atomFmt{\bang}}
\newcommand{\Act}{\mathord{\atomFmt{A}}}
\newcommand{\ActIn}{\mathord{\Act^{\InTag}}}
\newcommand{\ActOut}{\mathord{\Act^{\OutTag}}}
\newcommand{\BLab}{\atomFmt{L}}

\newcommand{\atom}[2][]{\atomFmt{\ifempty{#1}{{\code{#2}}}{{\code{#2}}}}_{\atomColor{#1}}}
\newcommand{\atomA}{\atom{a}}
\newcommand{\atomB}{\atom{b}}
\newcommand{\atomIn}[2][]{{\atomFmt{\InTag}}\atom[#1]{#2}}
\newcommand{\atomOut}[2][]{{\atomFmt{\OutTag}}\atom[#1]{#2}}

\newcommand{\labL}[1][]{\atomFmt{\atomFmt{\ell}_{#1}}}
\newcommand{\labLi}[1][]{\atomFmt{\atomFmt{\ell}'_{#1}}}

\def\tsbColor{\color{Magenta}}
\newcommand{\tsbFmt}[1]{{\tsbColor{#1}}}
\newcommand{\tsbP}[1][]{\mathord{\tsbFmt{p}_{\tsbColor{#1}}}}
\newcommand{\tsbPi}[1][]{\mathord{\tsbP[#1]\tsbColor{'}}}
\newcommand{\tsbPii}[1][]{\mathord{\tsbP[#1]\tsbColor{''}}}

\newcommand{\tsbQ}[1][]{\mathord{\tsbFmt{q}_{\tsbColor{#1}}}}
\newcommand{\tsbQi}[1][]{\mathord{\tsbQ[#1]\tsbColor{'}}}
\newcommand{\tsbQii}[1][]{\mathord{\tsbQ[#1]\tsbColor{''}}}
\newcommand{\tsbQiii}[1][]{\mathord{\tsbQ[#1]\tsbColor{'''}}}

\newcommand{\tsbX}[1][]{\mathord{\tsbFmt{X}_{\tsbColor{#1}}}}

\newcommand{\tsbZ}[1][]{\mathord{\tsbFmt{Z}}}
\newcommand{\tsbW}[1][]{\mathord{\tsbFmt{W}}}

\newcommand{\recsym}{{\tsbColor{\operatorname{rec}}}}
\newcommand{\rec}[2]{\recsym\,{#1}.\,{#2}}

\newcommand{\success}{{\tsbColor{\mathbf{1}}}}

\def\clockColor{\color{Cerulean}}
\def\guardColor{\color{Cerulean}}
\def\resetColor{\color{Cerulean}}
\newcommand{\clockFmt}[1]{{\clockColor{#1}}}
\newcommand{\guardFmt}[1]{{\guardColor{#1}}}
\newcommand{\resetFmt}[1]{{\resetColor{#1}}}

\newcommand{\Clocks}{\mathbb{\clockColor{C}}}
\newcommand{\Val}{{\clockColor{\mathbb{V}}}}

\newcommand{\clockT}[1][]{\mathord{\clockFmt{t}_{\clockColor{#1}}}}
\newcommand{\clockTi}[1][]{{\clockT[#1]{\clockColor{'}}}}
\newcommand{\clockX}[1][]{\mathord{\clockFmt{x}_{\clockColor{#1}}}}

\newcommand{\clockN}[1][]{\mathord{\clockFmt{\nu}_{\clockColor{#1}}}}
\newcommand{\clockNi}[1][]{{\clockN[#1]{\clockColor{'}}}}
\newcommand{\clockNii}[1][]{{\clockN[#1]{\clockColor{''}}}}

\newcommand{\clockE}[1][]{\mathord{\clockFmt{\eta}_{\clockColor{#1}}}}
\newcommand{\clockEi}[1][]{{\clockE[#1]{\clockColor{'}}}}
\newcommand{\clockEii}[1][]{{\clockE[#1]{\clockColor{''}}}}
\newcommand{\clockEiii}[1][]{{\clockE[#1]{\clockColor{'''}}}}

\newcommand{\guardTrue}{\mathord{\guardFmt{\code{true}}}}

\newcommand{\guardG}[1][]{\mathord{\guardFmt{g}_{\guardColor{#1}}}}
\newcommand{\guardGi}[1][]{{\guardG[#1]{\guardColor{'}}}}

\newcommand{\resetR}[1][]{\mathord{\resetFmt{R}_{\resetColor{#1}}}}

\newcommand{\resetT}[1][]{\mathord{\resetFmt{T}_{\resetColor{#1}}}}

\newcommand{\sep}{\,\bnfmid\,}

\newcommand{\compile}[2]{\ifthenelse{\equal{#1}{yes}}{#2}{}}

\newcommand{\irule}[2]{
  \begin{array}{c}
    #1  \\ \hline
    #2
  \end{array}}

\newcommand{\bnfmid}{\;\big|\;}

\newcommand{\nrule}[1]{{\scriptsize \textsc{#1}}}

\newcommand{\sem}[1]{\mbox{\ensuremath{\llbracket{{#1}}\rrbracket}}}

\newcommand{\bind}[2]{\nicefrac{#2}{#1}}
\newcommand{\setenum}[1]{\{#1\}}
\newcommand{\setcomp}[2]{\left\{{#1} \,\mid\, {#2}\right\}}

\newcommand{\SumIntRaw}{\bigoplus}
\newcommand{\SumExtRaw}{\sum}

\newcommand{\sumInt}{\oplus}

\newcommand{\TSumInt}[4][]{\SumIntRaw_{#1} {#2}{\ifempty{#3}{}{\setenum{#3}}} \, . \, {#4}}
\newcommand{\TSumExt}[4][]{\SumExtRaw_{#1} {#2}{\ifempty{#3}{}{\setenum{#3}}} \, . \, {#4}}

\newcommand{\TsumI}[3]{{#1} \ifempty{#2}{}{\setenum{#2}} \ifempty{#3}{}{ . \, {#3}}}
\newcommand{\TsumE}[3]{\TsumI{#1}{#2}{#3}}

\newcommand{\compliant}[0]{\mathbin{\bowtie_s}}
\newcommand{\acompliant}[0]{\mathbin{\bowtie_a}}
\newcommand{\rcompliant}[0]{\overline{\mathbin{\bowtie}}}

\newcommand{\smove}[1]{\mathrel{\xrightarrow{#1}_{\mathrm{s}}}} 
\newcommand{\asmove}[1]{\mathrel{\xrightarrow{#1}_{\mathrm{a}}}} 

\newcommand{\kindK}[1][]{\clockFmt{\mathcal{K}}_{\clockColor{#1}}}
\newcommand{\kindKi}[1][]{\clockFmt{\mathcal{K}'}_{\clockColor{#1}}}

\newcommand{\hcn}[1][]{\def\optdub{#1}%
\ifx\optdub\empty\Phi\else\Phi({#1})\fi}
\newcommand{\recLevel}[1][]{\def\optdub{#1}%
\ifx\optdub\empty\mathrm{RL}\else\mathrm{RL}({#1})\fi}

\newcommand{\mqr}{\rho}
\newcommand{\mqri}{\rho'}
\newcommand{\mqrii}{\rho''}
\newcommand{\mqs}{\sigma}
\newcommand{\mqsi}{\sigma'}

\usepackage{tikz} 

\newcommand{\mytitle}{On Urgency in Asynchronous Timed Session Types}
\title{\mytitle\footnote{This work has been partially sponsored by EPSRC
EP/N035372/1.}
}
\author{Maurizio Murgia
\institute{School of Computing\\
University of Kent.}\\
Canterbury, Uk\\
\email{M.Murgia@kent.ac.uk}}
\def\titlerunning{\mytitle}
\def\authorrunning{M. Murgia}

\begin{document}
\maketitle


\begin{abstract}
  We study an urgent semantics of asynchronous timed session types,
  where input actions happen as soon as possible. We show
  that with this semantics we can recover to the timed setting
  an appealing property of untimed session types: 
  namely, deadlock-freedom is preserved when passing from synchronous to asynchronous
  communication.
\end{abstract}

\section{Introduction} \label{sec:introduction}

Session types are abstractions of communication protocols~\cite{TakeuchiHK94},
used to statically or dynamically check that distributed programs interact correctly.
The original binary synchronous theory has subsequently been extended in 
several directions: explicit support for multiparty protocols and 
choreographies \cite{Honda16jacm}, asynchronous communication through FIFO buffers 
\cite{Honda16jacm}, 
time \cite{Bocchi14concur,BartolettiCM17}, and others \cite{Dezani09wsfm}. 

In this paper, we start an investigation on the relationships between synchronous 
and asynchronous session types in the timed binary setting. A related study has been performed
in the untimed setting \cite{BSZ14concur}, where, among other things, 
it has been proved that deadlock freedom in
session types interacting with synchronous communication is
preserved if messages are buffered. As reasoning about
synchronous systems is easier (synchronous progress is decidibile, 
asynchronous one is not), concurrent applications 
can be designed and verified
with synchronous communication in mind, and then run on top of real-world 
asynchronous mediums (\eg, TCP) while
preserving correctness. We refer to this practice as 
\emph{design synchronous/deploy asynchronous} methodology.

In the timed setting, as noted in \cite{BartolettiCM17}, this property
is lost, at least with the asynchronous semantics of \cite{Bocchi14concur}.
In this work, we propose an alternative semantics of asynchronous session types,
that forbids delays when reading actions are possible, similar to
an urgent semantics of Communicating Timed Automata \cite{ctaRefinement}.  
As noted in \cite{ctaRefinement}, this semantics, 
that we call \emph{input urgent} asynchronous semantics, better captures   
common reading primitives of programming languages/APIs, 
that return as soon as a message is available. Our semantics makes therefore
session types abstract models of \emph{programs}.
The semantics of \cite{Bocchi14concur}, instead, being more general (allows more behaviour),
seems preferable for modelling \emph{protocols} at a higher level of abstraction.

The main contribution of this paper is that the preservation result of \cite{BSZ14concur}
can be lifted to the timed setting, when using input urgent semantics. 
As timed synchronous progress is decidable \cite{BartolettiCM17}, and, as discussed above, 
input urgent semantics models programs, our result 
paves the way for the application of the design synchronous/deploy asynchronous 
methodology to time-sensitive distributed software.

\section{Synchronous timed session types} \label{sec:tst}
We now introduce timed session types (TST), 
their synchronous semantics and the associated notion of progress.
The material of this section is taken from \cite{BartolettiCM17}, with minor variations. 
The clock based model of time is borrowed from Timed Automata \cite{Alur94theory}.

\paragraph{Preliminaries.}

Let $\Act$ be a set of \emph{actions},
ranged over by $\atomA, \atomB, \ldots$.
We denote with $\ActOut$ the set $\setcomp{\atomOut{a}}{\atomA \in \Act}$
of \emph{output actions},
with $\ActIn$ the set $\setcomp{\atomIn{a}}{\atomA \in \Act}$
of \emph{input actions},
and with $\BLab = \ActOut \cup \ActIn$ the set of \emph{branch labels},
ranged over by $\labL,\labLi,\ldots$.
\newcommand{\guardsyntax}{%
  \guardG 
  \; ::= \;
  \guardTrue 
  \ \sep \ 
  \neg \guardG 
  \ \sep \ 
  \guardG  \land \guardG  
  \ \sep \
  \clockT \circ d
  \ \sep \
  \clockT - \clockTi \circ d
}
We use $\delta, \delta', \ldots$ to range over the set $\Realpos$
of not-negative real numbers,
and $d,d',\ldots$ to range over 
the set of natural numbers $\Nat$.
Let $\Clocks$ be a set of
\emph{clocks}, variables in $\Realpos$, ranged over by
$\clockT, \clockTi, \ldots$.  We use $\resetR,\resetT,\ldots \subseteq
\Clocks$ to range over sets of clocks. The syntax of guards 
(ranged over by $\guardG,\guardGi,\hdots$) is:
\[\guardsyntax \tag*{where $\circ \in \setenum{<,\leq,=,\geq,>}$}.\]
We give meaning to guards in terms of \emph{clock valuations},
namely functions of type $\Clocks \rightarrow \Realpos$
which associate each clock with its value.
We denote with $\Val = \Clocks \rightarrow \Realpos$ 
the set of clock valuations 
(ranged over by $\clockN, \clockE, \ldots$), %
and with $\clockN[0]$ the valuation mapping each clock to zero.
%
We use $\kindK,\kindKi,\ldots$ to range over sets of clock valuations.
We write $\clockN + \delta$ for the valuation 
which increases $\clockN$ by $\delta$,
\ie, $(\clockN + \delta)(\clockT) = \clockN(\clockT) + \delta$.
For a set $\resetR \subseteq \Clocks$,
we write $\reset{\clockN}{\resetR}$ for the \emph{reset} 
of the clocks in $\resetR$, 
\ie,
\[
\reset{\clockN}{\resetR}(\clockT) = \begin{cases}
  0 & \text{if $\clockT \in \resetR$} \\
  \clockN(\clockT) & \text{otherwise}
\end{cases}
\]

\begin{definition}[\textbf{Semantics of guards}]
  Let $\guardG$ be a guard.
  We define the set of clock valuations $\sem{\guardG}$
  inductively as follows, where $\circ \in \setenum{<,\leq,=,\geq,>}$:
  \[
  \begin{array}{lcl}
    \sem{\guardTrue} = \Val
    &
    \sem{\neg\guardG} = \Val \setminus \sem{\guardG} 
    \qquad
    &
    \sem{\guardG[1] \land \guardG[2]} = \sem{\guardG[1]} \cap \sem{\guardG[2]} 
    \\[5pt]
    \sem{\clockT \circ d} = \setcomp{\clockN}{\clockN(\clockT) \circ d} 
    \quad
    & & 
    \sem{\clockT - \clockTi \circ d} = \setcomp{\clockN}{\clockN(\clockT) - \clockN(\clockTi) \circ d} 
  \end{array}
  \]
\end{definition}

\begin{definition}[\textbf{Past}]
Let $\kindK$ be a set of clock valuations. We define $\past{\kindK}$ (the past of $\kindK$)
as follows:
\[
\past{\kindK} = \setcomp{\clockN}{\exists \delta \geq 0 : \clockN+\delta \in \kindK}
\]
\end{definition}
To model messages in transit we use queues. Queues are terms of the following grammar:
\[
\mqr\; :: =  \; \emptyset \ \mid \  \atomA;\mqr \\
\]  
We use $\mqr,\mqs$ to range over queues and we omit trailing occurrences of $\emptyset$. 
We write $|\mqr|$ for the number of messages in the queue $\mqr$ 
(we omit the straigthforward definition).

\paragraph{Syntax.}

A TST $\tsbP$ models the behaviour of a single participant
involved in an interaction. Roughly, in an internal choice 
$\TSumInt[i]{\atomOut[i]{a}}{\guardG[i],\resetR[i]}{\tsbP[i]}$ 
a participant has to perform one of the outputs $\atomOut[i]{a}$ in a time window 
where $\guardG[i]$ is true.
%
Conversely, in an external choice
$\TSumExt[i]{\atomIn[i]{a}}{\guardG[i],\resetR[i]}{\tsbQ[i]}$
the participant is available to receive
each message $\atom[i]{a}$ in \emph{any instant} 
within the time window defined by $\guardG[i]$.

\begin{definition}
  \label{def:tst:syntax}
  \emph{Timed session types} $\tsbP,\tsbQ,\ldots$ 
  are terms of the following grammar:
  \begin{align*}
    \tsbP \;\; & ::= \;\;
    \success
    \ \sep \
    \TSumInt[i \in I]{\atomOut[i]{a}}{\guardG[i],\resetR[i]}{\tsbP[i]} 
    \ \sep  \
    \TSumExt[i \in I]{\atomIn[i]{a}}{\guardG[i],\resetR[i]}{\tsbP[i]}
    \ \sep \
    \rec \tsbX \tsbP
    \ \sep \
    \tsbX
  \end{align*}
  where   
  \begin{inparaenum}[(i)]
  \item $I \neq \emptyset$ and finite,
  \item actions in internal/external choices are pairwise distinct,
  \item recursion is guarded.
  \end{inparaenum}
  %
  We omit true guards, empty resets, 
  and trailing occurrences of $\success$.
\end{definition}

\paragraph{Synchronous semantics.}
Semantics of TSTs is given in terms of a timed labelled transition relation between
configurations (defined below). Labels (ranged over by $\alpha,\alpha',\hdots$) 
are either silent actions $\tau$, delays $\delta > 0$, or branch labels. Labels
$\delta$ model elapse of time. Branch labels and $\tau$ model discrete actions,
and take no time.

\begin{definition}{\bf(Configurations)}\label{def:configuration}
A \emph{configuration} is a term of the form 
$(\tsbP,\mqr,\clockN) \mid (\tsbQ,\mqs,\clockE)$.
A configuration is synchronous if $|\mqr|,|\mqs| \leq 1$.
\end{definition}

We are now ready to define the synchronous semantics of TSTs. 
Unlike \cite{BartolettiCM17}, that uses committed choices,
we use queues in synchronous configurations to simplify
the comparison of synchronous and asynchronous semantics.
The two semantics are equivalent.

\begin{definition}{\bf(Synchronous semantics of TSTs)} 
\label{def:tst:semantics}\label{def:rdy}  
  The semantics of TSTs is defined as the smallest labelled relation 
  between synchronous configurations closed under the
  rules in \Cref{fig:tst:s_semantics}. 
  As usual, we denote with $\smove{}^*$ the reflexive and transitive
  closure of the relation $\smove{}$.
\end{definition}

\begin{figure}[t]
  \[
  \begin{array}{c}
    \begin{array}{cll}
      {( {\TsumI{\atomOut{a}}{\guardG,\resetR}{\tsbP}} 
	\sumInt \tsbPi,\;\emptyset,\; \clockN)
        \;\smove{\tau}\;
        (\tsbP,\; \atomA,\;\reset{\clockN}{\resetR})
      }
      \hspace{20pt}
      & \text{if } \clockN \in \sem{\guardG} 
      & \nrule{[$\sumInt$]}
      \\[4pt]
      {(\tsbP,\; \atomA,\;\clockN)
        \;\smove{\atomOut{a}}\;
        (\tsbP,\; \emptyset,\;\clockN)
      }
      & 
      & \nrule{[\bang]}    
      \\[4pt]
      {(\TsumE{\atomIn{a}}{\guardG,\resetR}{\tsbP} + \tsbPi, \;\emptyset ,\;\clockN)
        \;\smove{\atomIn{a}}\;
        (\tsbP, \; \emptyset,\;\reset{\clockN}{\resetR})}
      & \text{if } \clockN \in \sem{\guardG} 
      & \nrule{[\qmark]}    
      \\[4pt]
      (\tsbP,\; \emptyset,\;\clockN)\smove{\; \delta \; }
	(\tsbP,\;\emptyset,\; \clockN+\delta)
      & \text{if } \clockN + \delta \in \rdy{\tsbP}
      & \nrule{[Del]} \\[4pt]
	\irule{(\tsbP\setenum{\bind{\tsbX}{\rec \tsbX \tsbP}},\;\mqr,\;\clockN) 
	\;\smove{\alpha}\; (\tsbPi,\;\mqri,\;\clockNi)}
	{(\rec \tsbX \tsbP,\;\mqr,\;\clockN) \;\smove{\alpha}\;
	(\tsbPi,\;\mqri,\;\clockNi)} 
      & & \nrule{[Rec]}  
    \end{array}   \\[36pt]
   \begin{array}{cl}
    \irule{(\tsbP,\mqr,\clockN) \smove{\; \tau \; } (\tsbPi,\mqri,\clockNi)}
    {(\tsbP,\mqr,\clockN) \mid (\tsbQ,\mqs,\clockE)\smove {\; \tau \;}
    (\tsbPi,\mqri,\clockNi) \mid (\tsbQ,\mqs,\clockE)}
    & \nrule{[S-$\sumInt$]}
    \\[13pt]
    \irule
    {(\tsbP,\mqr,\clockN) \smove{\; \delta \; } (\tsbP,\mqr,\clockNi) \quad 
      (\tsbQ,\mqs,\clockE) \smove{\; \delta \; } (\tsbQ,\mqs,\clockEi)}
    {(\tsbP,\mqr,\clockN) \mid (\tsbQ,\mqs,\clockE) \smove{\; \delta \;} 
      (\tsbP,\mqr,\clockNi) \mid (\tsbQ,\mqs,\clockEi)}
    & \nrule{[S-Del]}
    \\[13pt]
    \irule
    {(\tsbP,\mqr,\clockN) \smove{\; \atomOut{a} \; } (\tsbPi,\mqri,\clockNi) \quad 
      (\tsbQ,\mqs,\clockE) \smove{\; \atomIn{a} \; } (\tsbQi,\mqsi,\clockEi)}
    {(\tsbP,\mqr,\clockN) \mid (\tsbQ,\mqs,\clockE) \smove{\; \tau \;} 
      (\tsbPi,\mqri,\clockNi) \mid (\tsbQi,\mqsi,\clockEi)}
    &  \nrule{[S-$\tau$]}\\[13pt]
   \end{array}
    \\[36pt]
    \rdy{\TSumInt{\atomOut[i]{a}}{\guardG[i],\resetR[i]}{\tsbP[i]}} = 
    \past{\bigcup \sem{\guardG[i]}}
    \hspace{12pt}
    \rdy{\SumExtRaw \cdots} = \rdy{\success} = \Val\\[4pt]
    \hspace{12pt}
    \rdy{\rec \tsbX \tsbP} = \rdy{\tsbP\setenum{\bind{\tsbX}{\rec \tsbX \tsbP}}}
  \end{array}
  \]
  \caption{Semantics of synchronous timed session types (symmetric rules omitted).}
  \label{fig:tst:s_semantics}
\end{figure}

Rule~\nrule{[$\sumInt$]} allows to commit to the branch
$\atomA$ of an internal choice, when the corresponding
guard is satisfied in the clock valuation $\clockN$ and the queue is empty.
This results in the configuration
$(\tsbP,\; \atomA,\;\clockN)$ which 
can only fire $\atomOut{a}$ (\nrule{[\bang]}).
Rule~\nrule{[\qmark]} allows an external choice to fire any of its
enabled input actions. Note that the queue must be empty. This is
what makes the semantics synchronous. Indeed, without the emptyness
requirement, we would have obtained a 1-bounded asynchronous semantics.
Rule~\nrule{[Del]} allows time to pass;
this is always possible for external choices and success term $\success$, 
while for an internal choice we require, through the function $\rdy{}$,
that some guard remains satisfiable (now or in the future). 
Note that also here we require empty queues:
this guarantees that messages are read at the same time of writing, \ie communication
is synchronous.
The other rules are almost standard.

\begin{example}
  \label[example]{ex:tst:committed-choice}
  Let 
  \(
  \tsbP = 
  \TsumI{\atomOut{a}}{}{} \sumInt 
  \TsumI{\atomOut{b}}{\clockT \geq 2}{}
  \) and
  \(
  \tsbQ = 
  \TsumE{\atomIn{b}}{\clockT \geq 5}{}
  \). $\tsbP$ internally chooses whether to send $\atomA$ at any time or 
  $\atomB$ after a delay of at least 2 time units. $\tsbQ$ instead waits for a $\atomB$ message
  after a delay of 5 time units. 
  Three possible reductions of the composition of $\tsbP$ with $\tsbQ$ are the following:
  \begin{align}
    \nonumber
    (\tsbP,\emptyset,\clockN[0]) \mid (\tsbQ,\emptyset,\clockE[0]) 
    \; \smove{\; 7 \;} \smove{\; \tau \;} \;
    & (\success,\atomB,\clockN[0] + 7) \mid 
    (\tsbQ,\emptyset,\clockE[0] + 7)
    \\
    \label{eq:tst:committed-choice:0}
    \; \smove{\; \tau \;} \;
    & (\success,\emptyset,\clockN[0] + 7) \mid 
    (\success,\emptyset,\clockE[0] + 7)
    \\
    \label{eq:tst:committed-choice:1}
    (\tsbP,\emptyset,\clockN[0]) \mid (\tsbQ,\emptyset,\clockE[0]) 
    \; \smove{\; \delta \;} \smove{\; \tau \;} \;
    & (\success,\atomA,\clockN[0] + \delta) \mid 
    (\tsbQ,\emptyset,\clockE[0] + \delta)
    \\
    \label{eq:tst:committed-choice:2}
    (\tsbP,\emptyset,\clockN[0]) \mid (\tsbQ,\emptyset,\clockE[0])  
    \; \smove{\; 3 \;} \smove{\; \tau \;} \;
    & (\success,\atomB,\clockN[0] + 3) \mid 
    (\tsbQ,\emptyset,\clockE[0] + 3)
  \end{align}
  The computation in~\eqref{eq:tst:committed-choice:0} reaches success.
  In~\eqref{eq:tst:committed-choice:1},
  $\tsbP$ commits to the choice $\atomOut{a}$ after some delay $\delta$;
  at this point, time cannot pass,
  and no synchronisation is possible.
  In~\eqref{eq:tst:committed-choice:2},
  $\tsbP$ commits to $\atomOut{b}$ after $3$ time units;
  here, the rightmost endpoint would offer $\atomIn{b}$,
  --- but not in the time chosen by the leftmost endpoint.
\end{example}

\paragraph{Synchronous progress.} \label{sec:tst-compliance}
We recall the progress based notion of compliance of \cite{BartolettiCM17},
that we refer here as synchronous compliance.
TSTs $\tsbP$ and $\tsbQ$ are synchronous compliant when their composition never reaches
a deadlock state. 

\begin{definition}[Synchronous compliance] 
  \label{def:deadlock}
  \label{def:compliance}
  \label{def:success}
  We say that $(\tsbP,\mqr,\clockN) \mid (\tsbQ,\mqs,\clockE)$ is
  \emph{success} whenever $\tsbP = \success = \tsbQ$ and $\mqr = \emptyset = \mqs$.
  We say that $(\tsbP,\mqr,\clockN) \mid (\tsbQ,\mqs,\clockE)$ is \emph{s-stuck}
  whenever $(\tsbP,\mqr,\clockN) \mid (\tsbQ,\mqs,\clockE) \not\smove{\tau}$ and
  there is no $\delta$ such that $(\tsbP,\mqr,\clockN) \mid (\tsbQ,\mqs,\clockE) 
  \smove{\delta}\smove{\tau}$.
  We say that $(\tsbP,\mqr,\clockN) \mid (\tsbQ,\mqs,\clockE)$
  is \emph{s-deadlock} whenever
  $(i)$ $(\tsbP,\mqr,\clockN) \mid (\tsbQ,\mqs,\clockE)$ not success,
  and $(ii)$ $(\tsbP,\mqr,\clockN) \mid (\tsbQ,\mqs,\clockE)$ is s-stuck.
  We then write $(\tsbP,\clockN) \compliant (\tsbQ,\clockE)$
  whenever:
  \[
  (\tsbP,\emptyset,\clockN) \mid (\tsbQ,\emptyset,\clockE)
  \smove{}^* 
  (\tsbPi,\mqr,\clockNi) \mid (\tsbQi,\mqs,\clockEi)
  \quad \text{ implies } \quad 
  (\tsbPi,\mqr,\clockNi) \mid (\tsbQi,\mqs,\clockEi)
  \text{ not s-deadlock}
  \]
  We say that $\tsbP$ and $\tsbQ$ are \emph{synchronous compliant}
  whenever $(\tsbP,\clockN[0]) \compliant (\tsbQ,\clockE[0])$
  (in short, $\tsbP \compliant \tsbQ$).
\end{definition}

\begin{example}
  Let $\tsbP = \atomIn{a}\setenum{\clockT \leq 3} . \atomOut{b}\setenum{\clockT \leq 3}$.
  We have that $\tsbP$ is compliant with
  $\tsbQ = \atomOut{a}\setenum{\clockT \leq 2}.\atomIn{b}\setenum{\clockT \leq 3}$, 
  but it is not compliant with 
  $\tsbQi = \atomOut{a}\setenum{\clockT \leq 4}.\atomIn{b}\setenum{\clockT \leq 4}$.
\end{example}



\section{Input urgent asynchronous timed session types} \label{sec:a-tst}
In this section we introduce input urgent asynchronous semantics,
the associated notion of progress, and we show some relationships
with the synchronous semantics.

\paragraph{Input urgent asynchronous semantics.}
We now introduce the input urgent semantics of TSTs. Note that here we use
configurations (\Cref{def:configuration}), \ie queues can be unbounded.

\begin{definition}{\bf(Input urgent asynchronous semantics of TSTs)} 
\label{def:tst:a-semantics}  
  The input urgent asynchronous semantics of TSTs is defined as the smallest 
  labelled relation between
  configuration closed under the
  rules in \Cref{fig:tst:as_semantics}.
  As usual, we denote with $\asmove{}^*$ the reflexive and transitive
  closure of the relation $\asmove{}$.
\end{definition}

Rule~\nrule{[$\sumInt$]} allows to append the message
$\atomA$ to the queue, when the corresponding
guard is satisfied in the clock valuation $\clockN$.
Rule \nrule{[\bang]} just says that the message in the head of the queue
can be consumed by the communication partner.
Rule~\nrule{[\qmark]} allows an external choice to fire any of its
enabled input actions. 
Rule~\nrule{[Del]} allows time to pass;
this is always possible for external choices and success term, 
while for an internal choice we require, through the function $\rdy{}$,
that some guard remains satisfiable. 
Rule~\nrule{[S-Del]} allows time to pass for composite systems.
While the first two premises are standard, the third one is what
makes the semantics urgent: we require, through the predicate $\delta$-sync,
that elapsing of time does not prevent nor delay any possible communication.
The other rules are almost standard.

\begin{figure}[t]
  \[
  \begin{array}{c}
    \begin{array}{cll}
      {( {\TsumI{\atomOut{a}}{\guardG,\resetR}{\tsbP}} 
	\sumInt \tsbPi,\;\mqr,\; \clockN)
        \;\asmove{\tau}\;
        (\tsbP,\; \mqr;\atomA,\;\reset{\clockN}{\resetR})
      }
      \hspace{20pt}
      & \text{if } \clockN \in \sem{\guardG} 
      & \nrule{[$\sumInt$]}
      \\[4pt]
      {(\tsbP,\; \atomA;\mqr,\;\clockN)
        \;\asmove{\atomOut{a}}\;
        (\tsbP,\; \mqr,\;\clockN)
      }
      & 
      & \nrule{[\bang]}    
      \\[4pt]
      {(\TsumE{\atomIn{a}}{\guardG,\resetR}{\tsbP} + \tsbPi, \;\mqr ,\;\clockN)
        \;\asmove{\atomIn{a}}\;
        (\tsbP, \; \mqr,\;\reset{\clockN}{\resetR})}
      & \text{if } \clockN \in \sem{\guardG} 
      & \nrule{[\qmark]}    
      \\[4pt]
      (\tsbP,\; \mqr,\;\clockN)\asmove{\; \delta \; }
	(\tsbP,\;\mqr,\; \clockN+\delta)
      & \text{if } \clockN + \delta \in \rdy{\tsbP}
      & \nrule{[Del]} \\[4pt]
	\irule{(\tsbP\setenum{\bind{\tsbX}{\tsbP}},\;\mqr,\;\clockN) \;\asmove{\alpha}\;
	(\tsbPi,\;\mqri,\;\clockNi)}
	{(\rec \tsbX \tsbP,\;\mqr,\;\clockN) \;\asmove{\alpha}\;
	(\tsbPi,\;\mqri,\;\clockNi)} 
      & & \nrule{[Rec]}  
    \end{array}   \\[40pt]
   \begin{array}{cl}
    \irule{(\tsbP,\mqr,\clockN) \asmove{\; \tau \; } (\tsbPi,\mqri,\clockNi)}
    {(\tsbP,\mqr,\clockN) \mid (\tsbQ,\mqs,\clockE)\asmove {\; \tau \;}
    (\tsbPi,\mqri,\clockNi) \mid (\tsbQ,\mqs,\clockE)}
    & \nrule{[S-$\sumInt$]}
    \\[13pt]
    \irule
    {
    \begin{array}{c}
      (\tsbP,\mqr,\clockN) \asmove{\; \delta \; } (\tsbP,\mqr,\clockNi) \quad 
      (\tsbQ,\mqs,\clockE) \asmove{\; \delta \; } (\tsbQ,\mqs,\clockEi)\\
	\forall \delta' < \delta: (\tsbP,\mqr,\clockN) \mid (\tsbQ,\mqs,\clockE)
	\text{ not }\delta'-\text{sync}
    \end{array}}
    {(\tsbP,\mqr,\clockN) \mid (\tsbQ,\mqs,\clockE) \asmove{\; \delta \;} 
      (\tsbP,\mqr,\clockNi) \mid (\tsbQ,\mqs,\clockEi)}
    & \nrule{[S-Del]}
    \\[18pt]
    \irule
    {(\tsbP,\mqr,\clockN) \asmove{\; \atomOut{a} \; } (\tsbPi,\mqri,\clockNi) \quad 
      (\tsbQ,\mqs,\clockE) \asmove{\; \atomIn{a} \; } (\tsbQi,\mqsi,\clockEi)}
    {(\tsbP,\mqr,\clockN) \mid (\tsbQ,\mqs,\clockE) \asmove{\; \tau \;} 
      (\tsbPi,\mqri,\clockNi) \mid (\tsbQi,\mqsi,\clockEi)}
    &  \nrule{[S-$\tau$]}
   \end{array}\\[42pt]
   (\tsbP,\mqr,\clockN) \mid (\tsbQ,\mqs,\clockE)\; \delta-\text{sync} \Longleftrightarrow
   \exists \atomA: \begin{cases}
 	(\tsbP,\mqr,\clockN + \delta) \asmove{\atomOut{a}} \land 
	(\tsbQ,\mqs,\clockE + \delta) \asmove{\atomIn{a}} & \text{or}\\
	(\tsbP,\mqr,\clockN + \delta) \asmove{\atomIn{a}} \land 
	(\tsbQ,\mqs,\clockE + \delta) \asmove{\atomOut{a}}
	\end{cases}
  \end{array}
  \]
  \caption{Urgent semantics of asynchronous timed session types 
  (symmetric rules omitted).}
  \label{fig:tst:as_semantics}
\end{figure}

\begin{example}\label{ex:as-semantics}
 Let 
  \(
  \tsbP = 
  \TsumI{\atomOut{a}}{\clockT \leq 2}{\TsumI{\atomOut{b}}{\clockT \leq 3}{}} 
  \),
  \(
  \tsbQ = 
  \TsumE{\atomIn{a}}{\clockT \geq 4}{\TsumE{\atomIn{b}}{\clockT \geq 5}{}}
  \),
  a possible execution of the system is:
\[
 (\tsbP,\emptyset,\clockN[0]) \mid (\tsbQ,\emptyset,\clockE[0]) \asmove{\tau}
 \asmove{\tau}(\success,\atomA;\atomB,\clockN[0]) \mid (\tsbQ,\emptyset,\clockE[0])
\asmove{4}\asmove{\tau}\asmove{1}\asmove{\tau}
\]
Where the for $\tau$ actions represent, respectively, an output of $\atomA$,
an output of $\atomB$, an input of $\atomA$, and an input of $\atomB$.
Note that urgency prevents transitions 
$(\success,\atomA;\atomB,\clockN[0]) \mid (\tsbQ,\emptyset,\clockE[0])
\asmove{\delta}$ if $\delta > 4$.
Let $\tsbQi = \TsumE{\atomIn{a}}{\clockT > 4}{\TsumE{\atomIn{b}}{\clockT \geq 5}{}}$,
i.e. $\tsbQi$ is like $\tsbQ$ but the constraint $\clockT \geq 4$ is substituted with
$\clockT > 4$. Message $\atomA$ cannot be consumed anymore: 
\[
 (\tsbP,\emptyset,\clockN[0]) \mid (\tsbQ,\emptyset,\clockE[0]) \asmove{\tau}
 \asmove{\tau}\asmove{4}
(\success,\atomA;\atomB,\clockN[0] + 4) \mid (\tsbQ,\emptyset,\clockE[0] + 4)
\]
Configuration $(\success,\atomA;\atomB,\clockN[0] + 4) \mid (\tsbQ,\emptyset,\clockE[0] + 4)$
cannot read $\atomA$ (because $\clockN[0] + 4 (\clockT) = 4$ and $4 \not> 4$),
and cannot delay, because for any $\delta$ there is a $\delta' < \delta$ such that
$(\success,\atomA;\atomB,\clockN[0] + 4) \mid (\tsbQ,\emptyset,\clockE[0] + 4)$ 
is $\delta'-$sync. Problems like that are well-known when dealing with urgency 
\cite{BornotST97}, therefore it is usually assumed that there is a first instant in which an 
urgent action becames enabled. In our setting, this assumption corresponds to forbid guards 
in the form $\clockX > n$. We do not make this assumption here just because our result 
does not rely on it.
\end{example}

\paragraph{Synchrony vs asynchrony.}
We remark some differences between the semantics in 
\Cref{fig:tst:s_semantics,fig:tst:as_semantics}:
\begin{itemize}
\item Synchronous semantics is defined on synchronous configurations, 
namely buffers are 1-bounded;
asynchronous semantics is defined on configurations, where buffers are unbounded.
\item With synchronous semantics, by rule \nrule{[Del]} of \Cref{fig:tst:s_semantics}, 
time can pass only if all buffer are empty. With asynchronous semantics,
time may pass even if buffers are not empty. However, by rule \nrule{[S-Del]} of
\Cref{fig:tst:as_semantics}, not empty buffers still constrain time passing: 
message comsumption is never delayed.
\end{itemize}

\begin{example}
Let $\tsbP$ and $\tsbQ$ be as in \Cref{ex:as-semantics}.
Intuitively, their composition should not succeed with synchronous semantics,
as $\tsbP$ writes $\atomA$ strictly earlier then when $\tsbQ$ is going to read it.
A possible complete execution is:
\[
(\tsbP,\emptyset,\clockN[0]) \mid (\tsbQ,\emptyset,\clockE[0]) \smove{\tau}
(\tsbP,\atomA,\clockN[0]) \mid (\tsbQ,\emptyset,\clockE[0])
\]
Where $(\tsbP,\atomA,\clockN[0]) \mid (\tsbQ,\emptyset,\clockE[0])$ is
s-deadlock: it is not success, it cannot perform actions, and it cannot delay 
(one buffer is not empty).
\end{example}

Below, we sketch a proof of some relations between synchronous and asynchronous semantics.
Namely, asynchronous semantics simulates the synchronous one. Furthermore,
when queues are empty, asynchronous delays are mimicked by the synchronous semantics.
\begin{lemma}\label{lem:async-simulates-sync}
Let $(\tsbP,\mqr,\clockN) \mid (\tsbQ,\mqs,\clockE)$ be a synchronous configuration.
Then:
\[
(\tsbP,\mqr,\clockN) \mid (\tsbQ,\mqs,\clockE) \smove{\alpha} 
(\tsbPi,\mqri,\clockNi) \mid (\tsbQi,\mqsi,\clockEi) \implies
(\tsbP,\mqr,\clockN) \mid (\tsbQ,\mqs,\clockE) \asmove{\alpha} 
(\tsbPi,\mqri,\clockNi) \mid (\tsbQi,\mqsi,\clockEi)
\]
Furthermore:
\[
(\tsbP,\emptyset,\clockN) \mid (\tsbQ,\emptyset,\clockE) \asmove{\delta} 
(\tsbPi,\mqri,\clockNi) \mid (\tsbQi,\mqsi,\clockEi) \implies
(\tsbP,\emptyset,\clockN) \mid (\tsbQ,\emptyset,\clockE) \smove{\delta} 
(\tsbPi,\mqri,\clockNi) \mid (\tsbQi,\mqsi,\clockEi)
\]
\end{lemma}
\begin{proof}
First note the following facts (can be easily proved by rule induction):
\begin{equation}\label{lem:async-simulates-sync-eq1}
\forall (\tsbP,\mqr,\clockN),(\tsbPi,\mqri,\clockNi):\;\;
(\tsbP,\mqr,\clockN) \smove{\alpha} (\tsbPi,\mqri,\clockNi)\; \implies\;
(\tsbP,\mqr,\clockN) \asmove{\alpha} (\tsbPi,\mqri,\clockNi)
\end{equation}
\begin{equation}\label{lem:async-simulates-sync-eq2}
\forall (\tsbP,\mqr,\clockN):\;\;
(\tsbP,\mqr,\clockN) \smove{\delta}\; \implies
\mqr =\emptyset
\end{equation}
\begin{equation}\label{lem:async-simulates-sync-eq3}
\forall \tsbP,\clockN,\tsbPi,\clockNi:\;\;
(\tsbP,\emptyset,\clockN) \asmove{\delta} (\tsbPi,\emptyset,\clockNi)\; \implies\;
(\tsbP,\emptyset,\clockN) \smove{\delta} (\tsbPi,\emptyset,\clockNi)
\end{equation}
Back to the main statement, the first part can be proved by cases on the rule used
in the derivation of $(\tsbP,\mqr,\clockN) \mid (\tsbQ,\mqs,\clockE) \smove{\alpha} 
(\tsbPi,\mqri,\clockNi) \mid (\tsbQi,\mqsi,\clockEi)$.
We only show the more complicated case, namely rule \nrule{[S-Del]}. Suppose:
\[
\irule
    {(\tsbP,\mqr,\clockN) \smove{\; \delta \; } (\tsbP,\mqr,\clockNi) \quad 
      (\tsbQ,\mqs,\clockE) \smove{\; \delta \; } (\tsbQ,\mqs,\clockEi)}
    {(\tsbP,\mqr,\clockN) \mid (\tsbQ,\mqs,\clockE) \smove{\; \delta \;} 
      (\tsbP,\mqr,\clockNi) \mid (\tsbQ,\mqs,\clockEi)}
\]
By \Cref{lem:async-simulates-sync-eq2}: $\mqs = \emptyset = \mqr$.
Therefore, by an inspection of the rules in \Cref{fig:tst:as_semantics},
we can conclude that, for all $\delta'$ and for all $\atomA$, 
both $(\tsbP,\mqr,\clockN) \not\asmove{\; \atomOut{a} \; }$ and
$(\tsbQ,\mqs,\clockE) \not\asmove{\; \atomOut{a} \; }$. So, 
$(\tsbP,\mqr,\clockN) \mid (\tsbQ,\mqs,\clockE)$ not $\delta'-$sync
for any $\delta'$. Therefore, thanks to \Cref{lem:async-simulates-sync-eq1},
we can use rule \nrule{[S-Del]}:
\[
\irule{(\tsbP,\mqr,\clockN) \asmove{\; \delta \; } (\tsbP,\mqr,\clockNi) \quad 
      (\tsbQ,\mqs,\clockE) \asmove{\; \delta \; } (\tsbQ,\mqs,\clockEi)
	\quad \forall \delta' < \delta: (\tsbP,\mqr,\clockN) \mid (\tsbQ,\mqs,\clockE)
	\text{ not }\delta'-\text{sync}}
    {(\tsbP,\mqr,\clockN) \mid (\tsbQ,\mqs,\clockE) \asmove{\; \delta \;} 
      (\tsbP,\mqr,\clockNi) \mid (\tsbQ,\mqs,\clockEi)}
\]

For the furthermore case, the only possibility is:
\[
\irule
    {(\tsbP,\emptyset,\clockN) \asmove{\; \delta \; } (\tsbP,\emptyset,\clockNi) \quad 
      (\tsbQ,\emptyset,\clockE) \asmove{\; \delta \; } (\tsbQ,\emptyset,\clockEi)
	\quad \forall \delta' < \delta: 
(\tsbP,\emptyset,\clockN) \mid (\tsbQ,\emptyset,\clockE) \text{ not }\delta'-\text{sync}}
    {(\tsbP,\emptyset,\clockN) \mid (\tsbQ,\emptyset,\clockE) \asmove{\; \delta \;} 
      (\tsbP,\emptyset,\clockNi) \mid (\tsbQ,\emptyset,\clockEi)}
     \nrule{[S-Del]}
\]
By \Cref{lem:async-simulates-sync-eq3}:
\[
    \irule
    {(\tsbP,\emptyset,\clockN) \smove{\; \delta \; } (\tsbP,\emptyset,\clockNi) \quad 
      (\tsbQ,\emptyset,\clockE) \smove{\; \delta \; } (\tsbQ,\emptyset,\clockEi)}
    {(\tsbP,\emptyset,\clockN) \mid (\tsbQ,\emptyset,\clockE) \smove{\; \delta \;} 
      (\tsbP,\emptyset,\clockNi) \mid (\tsbQ,\emptyset,\clockEi)}
    \nrule{[S-Del]}
\]
\end{proof}

\paragraph{Asynchronous progress.} \label{sec:tst-acompliance}
We extend the notion of progress to the asynchronous setting. 
It differs from \Cref{def:compliance}
only in that it uses the asynchronous semantics.

\begin{definition}[Asynchronous compliance] 
  \label{def:adeadlock}
  \label{def:acompliance}
  We say that $(\tsbP,\mqr,\clockN) \mid (\tsbQ,\mqs,\clockE)$ is \emph{a-stuck}
  whenever $(\tsbP,\mqr,\clockN) \mid (\tsbQ,\mqs,\clockE) \not\asmove{\tau}$ and
  there is no $\delta$ such that $(\tsbP,\mqr,\clockN) \mid (\tsbQ,\mqs,\clockE) 
  \smove{\delta}\asmove{\tau}$.
  We say that $(\tsbP,\mqr,\clockN) \mid (\tsbQ,\mqs,\clockE)$
  is \emph{a-deadlock} whenever
  $(i)$ $(\tsbP,\mqr,\clockN) \mid (\tsbQ,\mqs,\clockE)$ not success,
  and $(ii)$ $(\tsbP,\mqr,\clockN) \mid (\tsbQ,\mqs,\clockE)$ is a-stuck..
  We then write $(\tsbP,\clockN) \acompliant (\tsbQ,\clockE)$
  whenever:
  \[
  (\tsbP,\emptyset,\clockN) \mid (\tsbQ,\emptyset,\clockE)
  \asmove{}^* 
  (\tsbPi,\mqr,\clockNi) \mid (\tsbQi,\mqs,\clockEi)
  \quad \text{ implies } \quad 
  (\tsbPi,\mqr,\clockNi) \mid (\tsbQi,\mqs,\clockEi)
  \text{ not a-deadlock}
  \]
  We say that $\tsbP$ and $\tsbQ$ are \emph{asynchronous compliant}
  whenever $(\tsbP,\clockN[0]) \acompliant (\tsbQ,\clockE[0])$
  (in short, $\tsbP \acompliant \tsbQ$).
\end{definition}

\section{Results} \label{sec:tst-results}
In this section we sketch a proof of the main result of the paper 
(\Cref{th:scompliant-impies-acompliant}), namely that synchronous progress implies
asynchronous progress.
The proof is quite standard: we introduce a property (being the composition of
\emph{r-compliant} TSTs, \Cref{def:coind-compliance-async})
that is enjoyed by $(\tsbP,\emptyset,\clockN[0]) \mid (\tsbQ,\emptyset,\clockE[0])$, 
provided $\tsbP \compliant \tsbQ$.
We then show that r-compliance is preserved by transitions 
(\Cref{lem:rcompliance-preservation})
and that configurations of r-compliant TSTs are not a-deadlock.
(\Cref{lem:rcompliance-deadlock-free}).

R-compliance is defined below. It is based on the notion of \emph{reminder}, that, given a 
configuration and a queue, 
returns the configuration obtained after consuming the given queue immediately (without delays).
Note that the remainder is a partial operation, and it is not defined if the queue cannot be 
consumed, or some delay is required.
Then, r-compliance requires that:
\begin{itemize}
\item Queues can be consumed immediately.
\item The resulting configuration is composed by synchronous compliant TSTs.
\end{itemize}

\begin{definition}
\label{def:coind-compliance-async}
\label{def:remainder}
 We define the remainder of $(\tsbP,\mqr,\clockN)$ and queue $\mqs$, in symbols 
$(\tsbP,\mqr,\clockN) - \mqs$, inductively as follows:
\[
\begin{array}{rcll}
(\tsbP,\mqr,\clockN) - \emptyset & = & (\tsbP,\mqr,\clockN)\\
(\tsbP,\mqr,\clockN) - \atom{a};\mqs & = & (\tsbPi,\mqri,\clockNi) - \mqs & 
\quad\text{ if }(\tsbP,\mqr,\clockN) \asmove{\atomIn{a}} (\tsbPi,\mqri,\clockNi) 
\end{array}
\]
We say that  $(\tsbP,\mqr,\clockN)$ is \emph{r-compliant} with $(\tsbQ,\mqs,\clockE)$
(in symbols $(\tsbP,\mqr,\clockN) \rcompliant (\tsbQ,\mqs,\clockE)$)
if, for some $\tsbPi,\clockNi,\tsbQi,\clockEi$:
\[
(\tsbP,\mqr,\clockN) - \mqs = (\tsbPi,\mqr,\clockNi) \;\land\;
(\tsbQ,\mqs,\clockE) - \mqr = (\tsbQi,\mqs,\clockEi) \;\land\;
(\tsbPi,\clockNi) \compliant (\tsbQi,\clockEi)
\] 
\end{definition}

The following auxiliary lemma says that the r-compliant configurations, under asynchronous 
semantics, never allow delays unless both the queues are empty.
\begin{lemma}\label{lem:asyn-compliance-delta}
If $(\tsbP,\mqr,\clockN)\; \rcompliant\; (\tsbQ,\mqs,\clockE)$ and
$(\tsbP,\mqr,\clockN) \mid (\tsbQ,\mqs,\clockE) \asmove{\delta}$, then
$\mqr = \emptyset = \mqs$.
\end{lemma}
\begin{proof}
Suppose $(\tsbP,\mqr,\clockN) \mid (\tsbQ,\mqs,\clockE) \asmove{\delta}$.
First note that the only appliable rule is $\nrule{[S-Del]}$.
We have to show $\mqr = \emptyset = \mqs$. Suppose, by contradiction, this is not the case,
and assume that, say, $\mqr = \atomA;\mqrii$ for some $\atomA,\mqrii$.
By rule \nrule{[$\sumInt$]}, $(\tsbP,\mqr,\clockN) \asmove{\atomOut{a}}$.
Since $(\tsbP,\mqr,\clockN) \rcompliant (\tsbQ,\mqs,\clockE)$, 
it follows that $(\tsbQ,\mqs,\clockE) - \atomA;\mqrii$ is defined.
Then, by \Cref{def:remainder}, it must be $(\tsbQ,\mqs,\clockE) \asmove{\atomIn{a}}$.
But then $(\tsbP,\mqr,\clockN) \mid (\tsbQ,\mqs,\clockE)$ is $0$-sync, and so
rule \nrule{[S-Del]} does not apply: contradiction.
\end{proof}

The following proposition states that r-compliance is preserved by asynchronous transitions. 
\begin{proposition}\label{lem:rcompliance-preservation}
Let $(\tsbP,\mqr,\clockN)\; \rcompliant\; (\tsbQ,\mqs,\clockE)$ and
$(\tsbP,\mqr,\clockN) \mid (\tsbQ,\mqs,\clockE) \asmove{\alpha} 
(\tsbPi,\mqri,\clockNi) \mid (\tsbQi,\mqsi,\clockEi)$. Then:
\[(\tsbPi,\mqri,\clockNi)\; \rcompliant\; (\tsbQi,\mqsi,\clockEi)\]
\end{proposition}
\begin{proof}
Since $(\tsbP,\mqr,\clockN) \rcompliant (\tsbQ,\mqs,\clockE)$,
there exist $\tsbPii,\clockNii,\tsbQii,\clockEii$ such that:
\[
(\tsbP,\mqr,\clockN) - \mqs = (\tsbPii,\mqr,\clockNii) \;\land\;
(\tsbQ,\mqs,\clockE) - \mqr = (\tsbQii,\mqs,\clockEii) \;\land\;
(\tsbPii,\clockNii) \compliant (\tsbQii,\clockEii)
\]  
We proceed by cases on the rule used. 
\begin{itemize}
\item \nrule{[S-$\sumInt$]}. It must be $(\tsbQ,\mqs,\clockE) = (\tsbQi,\mqsi,\clockEi)$,
$\mqri = \mqr;\atomA$ for some $\atomA$, and $(\tsbP,\mqr,\clockN) \asmove{\; \tau \; }
(\tsbPi,\mqri,\clockNi)$.  By an inspection of rules in \Cref{fig:tst:as_semantics}, 
we can conclude
that $(\tsbP,\mqr,\clockN) \not\asmove{\; \atomIn{b} \; }$ for all $\atomB$.
Then $\mqs = \emptyset$: otherwise, $(\tsbP,\mqr,\clockN) - \mqs$ would be undefined.
Therefore, $(\tsbP,\mqr,\clockN) = (\tsbPii,\mqr,\clockNii)$ and
$(\tsbPi,\mqri,\clockNi) - \mqs = (\tsbPi,\mqri,\clockNi)$. By a simple induction
on the length of $\mqr$, we can conclude $(\tsbQ,\mqs,\clockE) - \mqr;\atomA =
(\tsbQii,\mqs,\clockEii) - \atomA$. We have to show that 
$(\tsbQii,\mqs,\clockEii) - \atomA = (\tsbQiii,\mqs,\clockEiii)$ for some 
$\tsbQiii,\clockEiii$ such that $(\tsbPi,\clockNi) \compliant (\tsbQiii,\clockEiii)$.
Note that, since $(\tsbPii,\clockNii) \compliant (\tsbQii,\clockEii)$ and 
$(\tsbPii,\clockNii)$ writes $\atomA$, it must be (by lemmas A.2 and 3.6 of 
\cite{BartolettiCM17}, modulo minor notational differences)
$(\tsbQii,\mqs,\clockEii) \smove{\;\atomIn{a}\;} (\tsbQiii,\mqs,\clockEiii)$, 
with $(\tsbPi,\clockNi) \compliant (\tsbQiii,\clockEiii)$.
By \Cref{lem:async-simulates-sync-eq1} in the proof of \Cref{lem:async-simulates-sync}, 
$(\tsbQii,\mqs,\clockEii) \asmove{\;\atomIn{a}\;} (\tsbQiii,\mqs,\clockEiii)$.
Therefore, $(\tsbQii,\mqs,\clockEii) - \atomA = (\tsbQiii,\mqs,\clockEiii)$ with
$(\tsbPii,\clockNii) \compliant (\tsbQii,\clockEii)$, and we are done.
\item \nrule{[S-Del]}. By \Cref{lem:asyn-compliance-delta}, it follows
$\mqr = \emptyset = \mqs$, and by rule \nrule{[S-Del]} $\mqri = \emptyset = \mqsi$ as well. 
By \Cref{def:coind-compliance-async}, 
$(\tsbP,\clockN) \compliant (\tsbQ,\clockE)$. By \Cref{lem:async-simulates-sync} and 
\Cref{def:compliance},  it follows $(\tsbPi,\clockNi) \compliant (\tsbQi,\clockEi)$, and so, 
since $\mqri$ and $\mqsi$ are both empty,
$(\tsbPi,\mqri,\clockNi) \rcompliant (\tsbQi,\mqsi,\clockEi)$.
\item \nrule{[S-$\tau$]}. It must be $(\tsbP,\mqr,\clockN) \asmove{\; \atomOut{a} \; } 
(\tsbPi,\mqri,\clockNi)$ and
$(\tsbQ,\mqs,\clockE) \asmove{\; \atomIn{a} \; } (\tsbQi,\mqsi,\clockEi)$.
By a simple induction on the rules in \Cref{fig:tst:as_semantics}, we can
conclude $\mqr = \atomA;\mqri$, $\mqs = \mqsi$, $\tsbP = \tsbPi$ 
(up to unfolding of recursion), and $\clockN = \clockNi$. 
Since $(\tsbQ,\mqs,\clockE) \asmove{\; \atomIn{a} \; } (\tsbQi,\mqsi,\clockEi)$,
by \Cref{def:remainder} it follows that
$(\tsbQ,\mqs,\clockE) - \atom{a};\mqri = (\tsbQi,\mqs,\clockEi) - \mqri$.
Clearly, up to unfolding of recursion, 
$(\tsbP,\mqr,\clockN) - \mqs = (\tsbPi,\mqr,\clockN) - \mqsi$.
Therefore, since $(\tsbP,\mqr,\clockN) \rcompliant (\tsbQ,\mqs,\clockE)$ by assumption,
also $(\tsbPi,\mqri,\clockNi) \rcompliant (\tsbQi,\mqsi,\clockEi)$.
\end{itemize}
\end{proof}

The following proposition states a-deadlock freedom of r-compliant configurations.
\begin{proposition}\label{lem:rcompliance-deadlock-free}
If $(\tsbP,\mqr,\clockN)\; \rcompliant\; (\tsbQ,\mqs,\clockE)$, then
$(\tsbP,\mqr,\clockN) \mid (\tsbQ,\mqs,\clockE)$ is not a-deadlock.
\end{proposition}
\begin{proof}
We have two cases:
\begin{itemize}
\item $\mqr = \emptyset \land \mqs = \emptyset$.
Then, $(\tsbP,\mqr,\clockN) - \mqs = (\tsbP,\mqr,\clockN)$ and
$(\tsbQ,\mqs,\clockE) - \mqr = (\tsbQ,\mqs,\clockE)$, with
$(\tsbP,\clockN) \compliant (\tsbQ,\clockE)$.
Therefore, $(\tsbP,\mqr,\clockN) \mid (\tsbQ,\mqs,\clockE)$ is not s-deadlock.
By \Cref{def:deadlock}, if $\tsbP = \success = \tsbQ$, then 
$(\tsbP,\mqr,\clockN) \mid (\tsbQ,\mqs,\clockE)$ is success and therefore not a-deadlock. 
If it is not the case that $\tsbP = \success = \tsbQ$, by \Cref{def:deadlock} 
there is $\delta$ such that 
$(\tsbP,\mqr,\clockN) \mid (\tsbQ,\mqs,\clockE) \smove{\delta}\smove{\tau}$ 
(we omit the simpler case where the $\tau$ move is performed without delay).
Then, by \Cref{lem:async-simulates-sync}, 
$(\tsbP,\mqr,\clockN) \mid (\tsbQ,\mqs,\clockE) \asmove{\delta}\asmove{\tau}$.
Therefore $(\tsbP,\mqr,\clockN) \mid (\tsbQ,\mqs,\clockE)$ is not a-deadlock.
\item $\mqr \neq \emptyset \lor \mqs \neq \emptyset$.
We show only the case $\mqr \neq \emptyset$. The other is similar.
It must be $\mqr = \atomA;\mqri$ for some $\atomA$ and $\mqri$.
Since $(\tsbQ,\mqs,\clockE) - \mqr$ is defined, it must be
$(\tsbQ,\mqs,\clockE) \asmove{\atomIn{a}}$. Therefore, by rule \nrule{[S-$\tau$]},
$(\tsbP,\mqr,\clockN) \mid (\tsbQ,\mqs,\clockE) \asmove{\tau}$, and so
$(\tsbP,\mqr,\clockN) \mid (\tsbQ,\mqs,\clockE)$ is not a-deadlock.
\end{itemize}
\end{proof}

The main result follows.
\begin{theorem}\label{th:scompliant-impies-acompliant}
If $\tsbP \compliant \tsbQ$ then $\tsbP \acompliant \tsbQ$.
\end{theorem}
\begin{proof}
Let $\tsbP \compliant \tsbQ$, and assume 
$(\tsbP,\emptyset,\clockN_0) \mid (\tsbQ,\emptyset,\clockE_0) \asmove{}^*
(\tsbPi,\mqri,\clockNi) \mid (\tsbQi,\mqri,\clockEi)$. We have to show
$(\tsbPi,\mqri,\clockNi) \mid (\tsbQi,\mqri,\clockEi)$ is not a-deadlock.
First note that, since queues are empty and $\tsbP \compliant \tsbQ$,
it holds that $(\tsbP,\emptyset,\clockN_0) \rcompliant (\tsbQ,\emptyset,\clockE_0)$.
By \Cref{lem:rcompliance-preservation} and a simple induction 
on the length of the reduction, we can derive
$(\tsbPi,\mqri,\clockNi) \rcompliant (\tsbQi,\mqri,\clockEi)$. 
By \Cref{lem:rcompliance-deadlock-free},
$(\tsbPi,\mqri,\clockNi) \mid (\tsbQi,\mqri,\clockEi)$ is not a-deadlock.
\end{proof}

\section{Conclusions and related work} \label{sec:related-work}
Following \cite{BartolettiCM17}, we pursued a line of research
aimed at lifting key properties of session types to the timed setting.
We have shown that the interesting property of 
preservation of untimed synchronous progress when passing to asynchronous
semantics (discovered in \cite{BSZ14concur}), can be recovered using a 
certain urgent asynchronous semantics, that closely models 
realistic programming primitives.

Timed session types have been introduced in \cite{Bocchi14concur}, 
in the multiparty asynchronous version,
where they have been used to statically type check a timed $\pi$-calculus. Their
theory has subsequently been extended to dynamic verification \cite{NBY17}.
\cite{BartolettiCM17} introduced the binary and synchronous theory, subsequently 
applied in a contract-oriented middleware \cite{CO2} and the 
companion verification tool-chain \cite{AtzeiB16}.
\cite{Bocchi15concur} studies progress in the context of Communicating Timed Automata 
\cite{KrcalY06}.
Several works study urgency in timed systems, here we mention \cite{BornotST97}.
As far as we know, urgency in the context of asynchronous communication
has only been studied in \cite{ctaRefinement}, 
where the idea of modelling input primitives
with input urgency originates. 

\newpage

\bibliographystyle{eptcs}
\bibliography{main}


\end{document}